\documentclass[11pt]{article}

\def\UseBibLatex{1}

\makeatletter
\def\input@path{{styles/}}
\makeatother

\providecommand{\BibLatexMode}[1]{}
\providecommand{\BibTexMode}[1]{}

\renewcommand{\BibLatexMode}[1]{#1}
\renewcommand{\BibTexMode}[1]{}

\ifx\UseBibLatex\undefined%
  \renewcommand{\BibLatexMode}[1]{}
  \renewcommand{\BibTexMode}[1]{#1}
\fi

\BibLatexMode{%
   \usepackage[bibencoding=utf8,style=alphabetic,backend=biber]{biblatex}%
   \usepackage{sariel_biblatex}%
}

\usepackage[in]{fullpage}%
\usepackage{amsmath}%
\usepackage{amssymb}%
\usepackage[table]{xcolor}%

\usepackage[amsmath,thmmarks]{ntheorem}%

\usepackage{titlesec}%
\usepackage{xcolor}%
\usepackage{mleftright}%
\usepackage{xspace}%
\usepackage{graphicx}
\usepackage{hyperref}%
\usepackage[inline]{enumitem}
\usepackage{hyperref}%
\usepackage[ocgcolorlinks]{ocgx2}

\hypersetup{%
      unicode,
      breaklinks,%
      colorlinks=true,%
      urlcolor=[rgb]{0.25,0.0,0.0},%
      linkcolor=[rgb]{0.5,0.0,0.0},%
      citecolor=[rgb]{0,0.2,0.445},%
      filecolor=[rgb]{0,0,0.4},
      anchorcolor=[rgb]={0.0,0.1,0.2}%
}

\titlelabel{\thetitle. }%

\theoremseparator{.}%

\theoremstyle{plain}%
\newtheorem{theorem}{Theorem}[section]

\newtheorem{lemma}[theorem]{Lemma}

\theoremstyle{plain}%
\theoremheaderfont{\sf} \theorembodyfont{\upshape}%
\newtheorem*{remark:unnumbered}[theorem]{Remark}%
\newtheorem{definition}[theorem]{Definition}

\newtheorem{problem}[theorem]{Problem}

\theoremheaderfont{\em}%
\theorembodyfont{\upshape}%
\theoremstyle{nonumberplain}%
\theoremseparator{}%
\theoremsymbol{\myqedsymbol}%
\newtheorem{proof}{Proof:}%

\providecommand{\emphind}[1]{}%
\renewcommand{\emphind}[1]{\emph{#1}\index{#1}}

\definecolor{blue25emph}{rgb}{0, 0, 11}

\providecommand{\emphic}[2]{}
\renewcommand{\emphic}[2]{\textcolor{blue25emph}{%
      \textbf{\emph{#1}}}\index{#2}}

\providecommand{\emphi}[1]{}%
\renewcommand{\emphi}[1]{\emphic{#1}{#1}}

\definecolor{almostblack}{rgb}{0, 0, 0.3}

\providecommand{\emphw}[1]{}%
\renewcommand{\emphw}[1]{{\textcolor{almostblack}{\emph{#1}}}}%

\providecommand{\emphOnly}[1]{}%
\renewcommand{\emphOnly}[1]{\emph{\textcolor{blue25emph}{\textbf{#1}}}}

\newcommand{\myqedsymbol}{\rule{2mm}{2mm}}

\newcommand{\HLink}[2]{\hyperref[#2]{#1~\ref*{#2}}}
\newcommand{\HLinkSuffix}[3]{\hyperref[#2]{#1\ref*{#2}{#3}}}

\newcommand{\figlab}[1]{\label{fig:#1}}
\newcommand{\figref}[1]{\HLink{Figure}{fig:#1}}

\newcommand{\lemlab}[1]{\label{lemma:#1}}
\newcommand{\lemref}[1]{\HLink{Lemma}{lemma:#1}}%

\providecommand{\eqlab}[1]{}%
\renewcommand{\eqlab}[1]{\label{equation:#1}}

\providecommand{\remove}[1]{}%
\newcommand{\Set}[2]{\left\{ #1 \;\middle\vert\; #2 \right\}}

 \newcommand{\brc}[1]{\mleft\{ {#1} \mright\}}

\newlist{compactenumA}{enumerate}{5}%
\setlist[compactenumA]{topsep=0pt,itemsep=-1ex,partopsep=1ex,parsep=1ex,%
   label=(\Alph*)}%

\newlist{compactenuma}{enumerate}{5}%
\setlist[compactenuma]{topsep=0pt,itemsep=-1ex,partopsep=1ex,parsep=1ex,%
   label=(\alph*)}%

\newlist{compactenumI}{enumerate}{5}%
\setlist[compactenumI]{topsep=0pt,itemsep=-1ex,partopsep=1ex,parsep=1ex,%
   label=(\Roman*)}%

\newlist{compactenumi}{enumerate}{5}%
\setlist[compactenumi]{topsep=0pt,itemsep=-1ex,partopsep=1ex,parsep=1ex,%
   label=(\roman*)}%

\newlist{compactitem}{itemize}{5}%
\setlist[compactitem]{topsep=0pt,itemsep=-1ex,partopsep=1ex,parsep=1ex,%
   label=\ensuremath{\bullet}}%

\numberwithin{figure}{section}%
\numberwithin{table}{section}%
\numberwithin{equation}{section}%

\usepackage{amsmath,amssymb,epsfig}

\def\A{{\cal A}}
\def\B{{\cal B}}
\def\C{{\cal C}}
\def\I{{\cal I}}
\def\T{{\cal T}}
\def\SS{{\cal S}}
\def\eps{{\varepsilon}}

\newcommand{\Reals}{{\rm I\!\hspace{-0.025em} R}}
\def\CH{\mbox{CH}}

\providecommand{\remove}[1]{[[[ {\bf Removed text:} \\ {\small #1} \\ ]]]}

\newenvironment{dfn}{{\vspace*{1ex} \noindent \bf Definition }}{\vspace*{1ex}}

\newcommand{\seg}[1]{\overline{#1}}

\newcommand{\prob}[1]{{\text{\sc{#1}}}}

\newcommand{\tsum}{\prob{3sum}}
\newcommand{\tsump}{\prob{3sum'}}

   \newcommand{\ed}{\prob{EqDist}}
   \newcommand{\scp}{\prob{Seg}\-\prob{ContPnt}}
   \newcommand{\pct}{\prob{PolyCont}}
   \newcommand{\cpcr}{\prob{ConvPolyContRot}}
   \newcommand{\cpctr}{\prob{ConvPolyContRigid}\-\prob{Mot}}
   \newcommand{\osshdt}{\prob{1SideSeg}\-\prob{HausDist}}
   \newcommand{\shdt}{\prob{SegHausDist}}
   \newcommand{\cpct}{\prob{ConvPolyContTrans}}%

\newcommand{\etal}{\textit{et~al.}\xspace}

\newcommand{\reduction}[3]{#1 \lll_{#2} #3}
\newcommand{\eqreduction}[3]{#1 \; \equiv_{#2} \; #3}
\newcommand{\tshard}{\tsum-hard}

\newcommand{\seclab}[1]{\label{sec:#1}}
\newcommand{\secref}[1]{\HLink{Section}{sec:#1}}

\bibliography{3sum}

\title{%
   Polygon Containment and Translational Min-Hausdorff-Distance between Segment Sets are \tsum-Hard%
   \thanks{%
      Work on this paper by the first author has been supported by the U.S. Army Research Office under MURI Grant DAAH04-96-1-0013.  Work by the second author has been supported by a grant from the U.S.-Israeli Binational Science Foundation.  A 2-page abstract of this paper appeared in SODA~'99~\cite{bh-pctmh-99} and the journal version appeared in \cite{bh-pctmh-01}. %
   } %
}%
\author{ Gill Barequet\thanks{ Center for Geometric Computing, Dept.\ of Computer Science, Johns Hopkins University, Baltimore, MD 21218.  E-mail: {\tt barequet\symbol{'100}cs.jhu.edu} (this author is currently affiliated with the Faculty of Computer Science, The Technion---IIT, Haifa 32000, Israel) } \and Sariel Har-Peled\thanks{ School of Mathematical Sciences, Tel Aviv University, Tel Aviv 69978, Israel.  E-mail: {\tt sariel\symbol{'100}math.tau.ac.il} } } \date{}

\begin{document}

\date{September 4, 1998}%

\maketitle

\begin{abstract}
   The \tsum\ problem represents a class of problems conjectured
   to require $\Omega (n^2)$ time to solve, where $n$ is the size of the
   input.
   Given two polygons $P$ and $Q$ in the plane, we show that
   some variants of the decision problem, whether there exists a
   transformation of $P$ that makes it contained in $Q$, are \tshard.
   In the first variant $P$ and $Q$ are any simple polygons and the
   allowed transformations are translations only;
   in the second and third variants both polygons are convex and we allow
   either rotations only or any rigid motion.
   We also show that finding the translation in the plane that minimizes
   the Hausdorff distance between two segment sets is \tshard.
\end{abstract}

\section{Introduction}

In this paper we show that a few polygon containment and Hausdorff distance
problems are in the class of \tsum-hard problems, introduced by
Gajentaan and Overmars~\cite{go-copcg-95}.
The \tsum\ problem is to decide whether there exists a triple $a,b,c$
in a set of $n$ integers such that $a+b+c=0$. Currently, the fastest known
algorithms for this problem require $\Theta (n^2)$ time.

A problem is called \tsum-hard (or $n^2$-hard in the original notation
of~\cite{go-copcg-95}) if any instance of the \tsum\ problem can be
reduced to some instance (with a comparable size) of the other problem
in $o(n^2)$ time, where $n$ is the size of the input.
Thus, a subquadratic algorithm for any such problem will imply a
subquadratic algorithm for \tsum, which is widely conjectured
not to exist~\cite{go-copcg-95}.

In a similar fashion Barrera~\cite{b-foash-96} presents a class of problems
whose claimed complexity is $\Omega(n^2 \log n)$ (where $n$ is the size of
the problem), since the best known algorithm for solving its base problem
(sorting $X+Y$, where each of $X$ and $Y$ is a set of $n$ integers) runs in
$O(n^2 \log n)$ time.

The paper is organized as follows. \secref{prelim} introduces the required terminology and notation.  Then we gradually establish several \tsum-hardness results.  We first prove in \secref{interval:containment} that deciding whether there exists a translation of a set of intervals on the real line, that makes it contained in another set of intervals, is \tshard.  Using this result we then prove in \secref{polygon:containment} that deciding whether there exists a translation of a simple polygon, that makes it contained in another simple polygon in the plane, is also \tshard.  We also show that the latter problem remains \tshard\ when the polygons are convex and either rotations or rigid motions are allowed.  Finally, we show in \secref{hausdorff:distance} that computing the translation in the plane that minimizes the Hausdorff distance between two sets of segments is \tshard. The last reduction uses a solution-separation technique.  We conclude in \secref{concl} with a few open problems.

\section{\tsum-Hard Problems}
\seclab{prelim}

In this section we closely follow~\cite{go-copcg-95} and introduce more
formally the notion of a \tsum-hard
problem and the machinery used to prove such a hardness result.

\begin{definition}
    Given two problems \prob{pr1} and \prob{pr2} we say that
    \prob{pr1} is $f(n)$-solvable using \prob{pr2} if every
    instance of \prob{pr1} of size $n$ can be solved by using a
    constant number of instances of \prob{pr2} (of size $O (n)$)
    and $O (f(n))$ additional time. We denote this by
    $$
       \reduction{\prob{pr1}}{f(n)}{\prob{pr2}}.
    $$

    If $\reduction{\prob{pr1}}{f(n)}{\prob{pr2}}$ and
    $\reduction{\prob{pr2}}{f(n)}{\prob{pr1}}$ we say that \prob{pr1} and
    \prob{pr2} are $f(n)$-equivalent and denote this by
    $$
       \eqreduction{\prob{pr1}}{f(n)}{\prob{pr2}}.
    $$
\end{definition}

When \prob{pr1} is $f(n)$-solvable using \prob{pr2}, this means that
\prob{pr1} is ``easier'' than \prob{pr2} (provided that solving \prob{pr1}
requires $\Omega (f (n))$ time). That is, when $f(n)$ is
sufficiently small, lower bounds for \prob{pr1} carry over to
\prob{pr2} and upper bounds for \prob{pr2} hold for \prob{pr1}.

The following two problems are defined as
representatives of their hardness class:

\begin{problem}[\tsum{}]
    Given a set $S$ of $n$ integers, are there $a,b,c \in S$ with $a+b + c = 0$?
\end{problem}

\begin{problem}[\tsump{}]
    Given three sets of integers $A, B,$ and $C$, each of size $n$, are there $a \in A$, $b \in B$, and $c \in C$ with $a + b = c$?
\end{problem}

\begin{theorem}[Gajentaan and Overmars~\cite{go-copcg-95}]
    $\eqreduction{\tsump}{n}{\tsum}$.
\end{theorem}

\begin{definition}
    A problem \prob{pr} is \tsum-hard if
    $\reduction{\tsum}{f(n)}{\prob{pr}}$, where $f(n) = o(n^2)$.
\end{definition}

In the next three sections we show that several geometric problems in
one and two dimensions are \tshard.

\section{Interval Containment}
\seclab{interval:containment}

In this section we show that the following two problems are \tshard:

\begin{problem}[\ed{} (Equal Distance)]
   Given two sets $P$ and $Q$ of $n$ and $m = O (n)$ real
   numbers,\footnote{Naturally, the same problem but with {\em rational\/}
      numbers is also \tshard.
   } respectively, is there a pair $p_1, p_2 \in P$ and a pair
   $q_1, q_2 \in Q$ such that $p_1 - p_2 = q_1 - q_2$?
\end{problem}

\begin{problem}[\scp\ (Segments Containing Points)]
   Given a set $P$ of $n$ real numbers and a set $Q$ of $m = O (n)$
   pairwise-disjoint intervals of real numbers, is there a
   real number (translation) $v$ such that $P + v \subseteq Q$?
\end{problem}

In the next section we interpret $P$ and $Q$ as sets of points and
segments, respectively (or, more generally, two segment sets), on the real
line. We therefore call the second problem ``Segments Containing Points''.

\begin{theorem}
    $\reduction{\tsump}{n}{\ed}$ and $\reduction{\tsump}{n}{\scp}$.
\end{theorem}

\begin{proof}
    Let $(A,B,C)$ be an instance of \tsump, where $|A|=|B|=|C| = n$. Assume without loss of generality that $A \cup B \cup C \subseteq (0,1)$ (this may be achieved by performing affine transformations).  We create a set $P$ of $2n$ numbers by mapping each number $c_i \in C$ into a pair of numbers with difference $3 - c_i$.  In particular,
   \begin{equation*}
      P = \Set{ 100i, 100i + 3 - c_i }{c_i \in C, i = 1,\ldots, n}.
   \end{equation*}
   We also define the following set of numbers:
   \begin{equation*}
       Q = A \cup \Set{ 3-b }{b \in B }.
   \end{equation*}
   We show that $(P,Q)$ is a corresponding instance of the \ed\ problem.

   First assume that there exists a solution to the \tsump\ instance,
   that is, a triple $a \in A$, $b \in B$, and $c_i \in C$ such that
   $a + b = c_i$. Then there exists a solution to the \ed\ instance:
   $(100 i + 3 - c_i) - (100 i) = (3 - b) - a$.
   Now assume that there exists a solution to the \ed\ instance,
   that is, $p_1, p_2 \in P$ and $q_1, q_2 \in Q$ such that
   $p_1 - p_2 = q_1 - q_2$. Obviously, $p_1$ and $p_2$ must correspond
   to the same $c_i \in C$, otherwise the difference between them would
   be at least 97, while the difference between $q_1$ and $q_2$ is at most 3.
   Moreover, $|p_1 - p_2| > 1$, thus $q_1$ and $q_2$ cannot correspond to
   a pair of points of $A$ or a pair of points of $B$, in which case
   $|q_1 - q_2| < 1$.
   Thus, there are numbers $100 i, 100 i + 3 - c_i \in P$ (where $c_i \in C$)
   and $a, 3 - b \in Q$ (where $a \in A, b \in B$), such that
   $(100 i - 3 - c_i) - (100 i) = (3 - b) - a$. This implies
   $a + b = c_i$, a solution to the \tsump\ instance.
   Thus we have $\reduction{\tsump}{n}{\ed}$.

   It is easy to extend the \ed\ instance to an \scp\ instance by adding
   two new intervals to $Q$:
   $$
      Q' = [-100 (n-1), -94] \cup Q \cup [100, 100 (n-1) + 6].
   $$
   (Here, every original point $q \in Q$ is interpreted as the interval
   $[q, q] \in Q'$.)
   The length and location of the two additional intervals are chosen so
   that every interval can contain a translation of only $(n-1)$ (but not
   $n$) pairs of $P$, and that even the union of the two new intervals
   cannot contain any translation of $P$.

   Recall that a triple of points $a \in A, b \in B, c_i \in C$,
   for which $a + b = c_i$, is reflected by two members ($q_1$ and $q_2$)
   in the intersection $(P + a - 100 i) \cap Q$, and vice versa.
   First assume that there is a solution to the \scp\ instance. We show
   a solution to the \ed\ instance, which implies a solution to the
   corresponding \tsump\ instance.
   Each new interval is of length $100 n - 194$, so a single interval
   cannot cover a translation of $P$, whose extreme numbers are within
   distance at least $100 n - 98$. The gap between the two new intervals
   is of length 194, so even their union cannot cover any translation
   of $P$. Hence, in a valid translation,
   at least one of the original numbers of $Q$ coincides
   with a translated number of $P$. The latter number corresponds to
   some $c_i \in C$. Due to the definition of $Q'$, the translation of the
   other number of $P$ that corresponds to $c_i$ must also coincide with a
   number of the original set $Q$, establishing the solution to the \ed\
   instance.
   Conversely, assume that there exists a solution to the \tsump\
   instance, which implies a solution to the corresponding \ed\ instance.
   We show that the solution to the latter instance is also a solution to
   the \scp\ instance: a translation of the pair $p_1, p_2 \in P$ that
   corresponds to $c_i$ is covered by $q_1 = a, q_2 = 3 - b \in Q$, and
   the translations of all the other points of $P$ are easily seen
   to be covered by the two additional intervals.
   Thus we have $\reduction{\tsump}{n}{\scp}$.
\end{proof}

Note that we have shown that
\begin{equation*}
    \reduction{\tsump}{n}{\scp}
\end{equation*}
but not that $\reduction{\ed}{n}{\scp}$. This is because we are able to reduce to \scp\ instances only \ed\ instances created by reductions from \tsump\ instances.

\section{Polygon Containment}
\seclab{polygon:containment}

In this section we show that several polygon containment problems are \tshard.
The symbol $o$ is used hereafter to denote the origin of coordinates.

\begin{problem}[\cpct]
    Given two convex polygons $P$ and $Q$ in the plane with $n$ and $m = O (n)$ edges, respectively, is there a translation of $P$ that makes it contained in $Q$?
\end{problem}

Obviously, every instance of the \cpct\ problem can be solved in linear time.

\begin{problem}[\pct\ (Polygon Containment under Translation)]
    Given two simple polygons $P$ and $Q$ in the plane with $n$ and $m = O (n)$ edges, respectively, is there a translation of $P$ that makes it contained in $Q$?
\end{problem}

\begin{problem}[\cpcr{}]
    Given two
   convex polygons $P$ and $Q$ in the plane with $n$ and $m = O (n)$ edges,
   respectively, is there a rotation of $P$ around $o$ that makes it
   contained in $Q$?
\end{problem}

\begin{problem}[\cpctr{}]
    Given two convex polygons $P$ and $Q$ in the plane with $n$ and $m = O (n)$ edges, respectively, is there a rigid motion (translation and rotation) of $P$ that makes it contained in $Q$?
\end{problem}

Given a set $S$ of intervals on the real line, we denote by $I (S)$ the
smallest interval of real numbers that covers the entire set $S$.

\begin{theorem}
    $\reduction{\scp}{n \log n}{\pct}$.
\end{theorem}

\begin{proof}
   Let $(A,B)$ be an \scp\ instance, such that $|A|, |B| = O (n)$.
   Assume without loss of generality that $A$ and $B$ are embedded along
   the $x$-axis. We construct a \pct\ instance to which there is a
   solution if and only if there exists a solution to the original
   \scp\ instance.

   We convert an interval set $S$ into a ``comb''
   (see \figref{fat})
   \begin{figure}
      \centering
      \begin{tabular}{c}
        \includegraphics{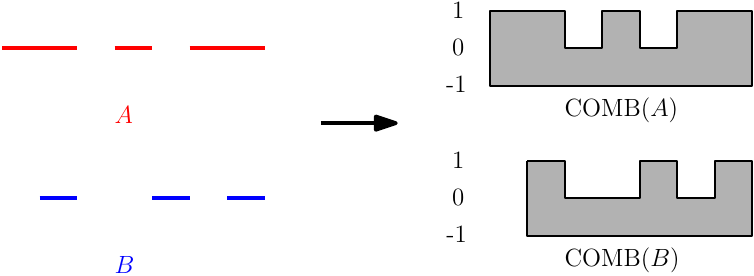}
      \end{tabular}
      \caption{Converting a segment set to a polygon}
      \figlab{fat}
   \end{figure}
   by mapping it to the simple polygon
   $\mbox{COMB}(S) = S \times [0,1] \cup I(S) \times [-1,0]$.
   The polygons $C_A = \mbox{COMB}(A)$ and $C_B = \mbox{COMB}(B)$ are
   computed in $O (n \log n)$ time by sorting $A$ and $B$ along the $x$-axis.
   Clearly, if $C_A + (u,v) \subseteq C_B$, for some $(u,v) \in \Reals^2$,
   then $v = 0$ (since the height of both $C_A$ and $C_B$ is 2).
   Moreover, it is easily seen (by definition) that
   $C_A + (u,0) \subseteq C_B$ if and only if $A + u \subseteq B$.~
\end{proof}

\begin{theorem}
    We have
    \begin{equation*}
        \reduction{\scp}{n \log n}{\cpcr}
    \end{equation*}
    and
    \begin{equation*}
       \reduction{\scp}{n \log n}{\cpctr}.
    \end{equation*}
\end{theorem}

\begin{proof}
   Let $(A',B')$ be an \scp\ instance such that $|A'|, |B'| = O (n)$.
   Assume without loss of generality that $A', B' \subseteq [0.45, 0.55]$.
   Let $A = \brc{0.1} \cup A' \cup \brc{0.9}$ and
   $B = [0,0.2] \cup B' \cup [0.8, 1]$.
   Clearly, $A + v \subseteq B$ if and only if
   $A'+v \subseteq B'$, for any $v \in \Reals$.

   Let $\C$ be the unit circle centered at $o$,
   $f (x) = (\sin (x / 100), \cos (x / 100) )$ be a mapping
   of the real line into $\C$, and $\A = f (A), \B = f (B)$.
   Obviously, there exists a translation $v \in \Reals$ such that
   $A + v \subseteq B$ if and only if there is a rotation of $\A$ around
   $o$ that makes $\A$ contained in $\B$.

   Given a circular interval (an arc) $\I \subset \C$, we denote by
   $p (\I)$ the intersection point between the two lines tangent to $\C$ at
   $l (\I)$ and $r (\I)$, the two endpoints of $\I$.
   (In case $\I$ is a point, $p (\I) = \I$.)
   For a set $\SS$ of arcs on $\C$, we define
   $W (\SS) = \CH (\brc{(0,0)} \cup
                   \bigcup_{\I \in \SS} \brc{l (\I), r (\I), p (\I)})$,
   where $\CH (\cdot)$ is the convex-hull operator.
   See \figref{circular:cpoly} for an illustration.
   \begin{figure}
      \centering
      \includegraphics{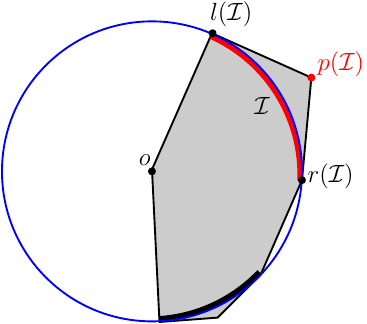}
      \caption{Constructing a convex polygon from a set of circular arcs.}
      \figlab{circular:cpoly}
   \end{figure}
   In fact, $\A$ and $\B$ occupy a very small portion of $\C$ due to the
   definition of $f$ that maps any segment in $A$ or $B$ to a very small
   arc on $\C$.
   Again, we compute $W (\A)$ and $W (\B)$ in $O (n \log n)$
   time by sorting $\A$ and $\B$ along $\C$.

   Given a solution to the original \scp\ instance, it immediately
   induces solutions to the corresponding \cpcr\ and \cpctr\
   problems: if $A + v \subseteq B$ (for some $v \in \Reals$), let $R_v$
   be the corresponding rotation around $o$ such that
   $R_v (\A) \subseteq (\B)$.
   By definition, $R_v (W (\A)) \subseteq W (\B)$.
   Similarly, a solution to the \cpcr\ instance induces a solution
   to the \scp\ instance.
   We thus have $\reduction{\scp}{n \log n}{\cpcr}$.

   We need more effort to show that a solution to the \cpctr\
   instance induces a solution to the \scp\ instance.
   Let, then, $(W (\A), W (\B))$ be a \cpctr\ instance to which
   there is a solution. Note that $W (\A)$ and $W (\B)$ are long skinny
   polygons (of length more than 1 and width less than 0.01\footnote{
      For ease of exposition we measure the ``width'' of the polygon in
      terms of the length of its defining arc.
   }), where one
   ``end'' of both polygons is the point $o$. Therefore the transformation
   that makes $W (\A)$ contained in $W (\B)$ maps $o$ either to a point
   near the origin (with a small or no rotation at all) or to a point near
   the other end of $W (\B)$, that is, $W (\A)$ is rotated around $o$ by
   roughly $180^\circ$ and translated by roughly $(0,1)$.

   Let us first rule out the latter case. The wider end of $W (\A)$ (made
   of the points of $\A$ on $\C$) is of width $(0.9-0.1)/100=0.008$, whereas
   the width of $W (\B)$ is $(1-0)/100=0.01$. A simple calculation shows
   that in order to have the wider end of $W (\A)$ contained in the narrower
   end of $W (\B)$, it should be translated upward (along the $y$-axis,
   after the rotation around the origin) by at least 1.2, thereby making
   $W (\A)$ ``penetrate'' the wider end of $W (\B)$ and disallowing this
   transformation.

   Consider, then, the former case that maps $o$ to a point
   close to it. The two extreme sides of $W (\A)$ (the segments
   $\seg{o \, f(0.1)}$ and $\seg{o \, f(0.9)}$)
   of length 1 are separated from the interior of $W (\A)$
   (lying between $f (\min A)$ and $f (\max A)$) by two areas of length less
   than 1. Therefore the transformation that makes $W (\A)$ contained in
   $W (\B)$ must map the two extreme sides of $W (\A)$ to lie inside
   $W ([f(0),f(0.2)])$ and $W ([f(0.8),f(1)])$. The transformed $W (\A)$
   can now be continuously and rigidly moved further to make the origin
   invariant under the compound transformation, while the image of
   $W (\A)$ is still contained in $W (\B)$. However, since now $o$ is
   invariant, the compound transformation is nothing but a
   rotation. Recall that we have already shown above that
   there exists a solution to the \cpcr\
   instance if and only if there exists a solution to the
   \scp\ instance. Therefore $\reduction{\scp}{n \log n}{\cpctr}$.
\end{proof}

The \cpcr\ problem can be solved in $O (n^2 \log n)$ time by performing an
angular sweeping of $P$ around the origin. During the rotation we maintain
the intersection points between $Q$ and the rotated $P$.
The events of this algorithm are rotation angles for which an intersection
point between $P$ and $Q$ identifies with some vertex of the polygons.
The goal is to track an event (if it exists) in which the rotated copy
of $P$ becomes fully contained in $Q$.

\section{Hausdorff Distance between Segment Sets}
\seclab{hausdorff:distance}

In this section we show that computing the minimum Hausdorff distance
under translation between two planar sets of line segments is \tshard.
Note that we allow segments to be ``infinite'', that is, lines.

\begin{definition}
   Let $A$ and $B$ be two compact subsets of $\Reals^2$.
   The {\em one-sided Hausdorff distance} from $A$ to $B$ is defined as
   $$
      \delta'(A,B) = \max_{a \in A} \min_{b \in B} ||a - b||
   $$
   (where $||\cdot||$ is the Euclidean distance),
   and the {\em two-sided Hausdorff distance} (``Hausdorff distance''
   for short) between $A$ and $B$ is defined as
   $$
      \delta(A,B) = \max (\delta'(A,B), \delta'(B,A)).
   $$

   The {\em minimum one-sided Hausdorff distance under translation}
   between $A$ and $B$ is
   $$
      d_H'(A,B) = \min_{T \in \Reals^2} \delta'(T(A),B).
   $$
   Similarly, the {\em minimum Hausdorff distance under translation}
   between $A$ and $B$ is
   $$
      d_H(A,B) = \min_{T \in \Reals^2} \delta(T(A),B).
   $$
   The translation $T$ that realizes the minimum distance is called the
   optimal translational matching between $A$ and $B$.
\end{definition}

\begin{problem}[\osshdt{}: 1-Sided Segment Hausdorff Distance under Translation]
    Given \linebreak two sets $A, B$, each containing $n$ segments in $\Reals^2$, compute $d_H'(A,B)$.
\end{problem}

The \osshdt\ problem is obviously \tshard\ since it is
a generalization of \scp: Let $(P,Q)$ be an \scp\ instance;
then, by definition, there exists a number $u \in \Reals$ for which
$P + u \subseteq Q$ if and only if $d_H'(P,Q) = 0$.

\begin{problem}[\shdt\ (Segment Hausdorff Distance under Translation)]
    Given two sets $A, B$, each containing $n$ segments in $\Reals^2$, compute $d_H(A,B)$.
\end{problem}

\begin{theorem}
    $\reduction{\scp}{n}{\shdt}$.
\end{theorem}

\begin{proof}
   We begin with a reduction of an \scp\ instance to an instance of \osshdt,
   and show how to extend the latter instance into an instance of \shdt.
   First we establish the following:

   \begin{lemma}
       \lemlab{1shd-res}%
       Let $(P,Q)$ be an \scp\ instance of size $n$. Then, either $d_H'(P,Q) = 0$ or $d_H'(P,Q) \geq \eps = \eps(P,Q)$, where $\eps > 0$ can be computed in only $O (n)$ time (without actually solving the corresponding \osshdt\ instance).
   \end{lemma}

   \begin{proof}
       First we note that the input to a \tsump\ problem is a set of integers, and that the reduction from \tsump\ to \scp\ (see \secref{interval:containment}) involves only rational operators. Therefore we may assume that the description of the \scp\ problem (numbers in $P$ and interval endpoints in $Q$) consists of only rational numbers, all of which belong to the interval $[0,1]$. Denote the largest denominator of all these numbers by $M$.  This number can be computed in $O (n)$ time.

      We set a lower bound on $d_H'(P,Q)$ in terms of $M$,
      under the assumption that $d_H'(P,Q) > 0$.
      Let $u$ be the translation that realizes $d_H'(P,Q)$.
      Clearly, there are two points $p_1, p_2 \in P$ and two respective
      points $q_1, q_2 \in Q$ (segment endpoints), for which
      $|(p_1+u)-q_1| = |(p_2+u)-q_2| = d_H'(P,Q)$.  Moreover,
      $p_1$ and $p_2$ are ``interlocked'' with $q_1$ and $q_2$,
      in the sense that modifying $u$ slightly (by increasing
      or decreasing it) results in enlarging either
      $|(p_1+u)-q_1|$ or $|(p_2+u)-q_2|$.  (If no such four
      numbers existed, we would be able to improve
      (decrease) the 1-sided Hausdorff distance under translation
      between $P$ and $Q$ by modifying slightly the value of $u$.)
      \figref{points:config}
      \begin{figure}
         \centering
         \includegraphics{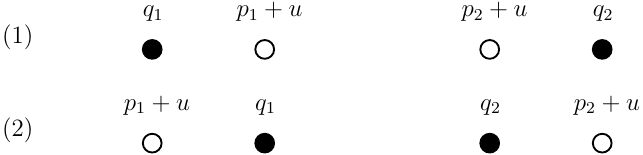}
         \caption{Two configurations that realize the
                  minimum Hausdorff distance under translation.}
         \figlab{points:config}
      \end{figure}
      shows the two possible configurations of $p_1, p_2$ (shown as empty
      circles), and $q_1, q_2$ (shown as filled circles).
      Without loss of generality we assume the upper configuration.
      Thus, $d_H'(P,Q) = (p_1 + u) - q_1 = q_2 - (p_2 + u)$, hence
      $u = (q_1 + q_2 - p_1 - p_2) / 2$ and
      $d_H'(P,Q) = (p_1 - p_2 - q_1 + q_2) / 2$.
      Therefore either $d_H'(P,Q) = 0$ or
      $d_H'(P,Q) \geq \eps(P,Q) = 1/(2 M^4)$.
   \end{proof}

   We use this separation property of the one-sided Hausdorff distance
   for extending the \osshdt\ instance into a (two-sided distance)
   \shdt\ instance.

   We define the four lines $\ell_1$:~$y = 1.6 \eps$, $\ell_2$:~$y = 0.8 \eps$, $\ell_3$:~$y = -0.8 \eps$, and $\ell_4$:~$y = -1.6 \eps$, and construct an \shdt\ instance $(A,B)$, where $A = P \cup \ell_2 \cup \ell_3$ and $B = Q \cup \ell_1 \cup \ell_4$ (see \figref{configuration}).
   \begin{figure}
      \centering
      \includegraphics{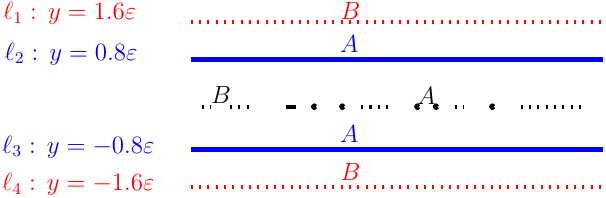}
      \caption{The construction that shows that \shdt\ is \tshard.}
      \figlab{configuration}
   \end{figure}
   In order to complete the proof
   we need to show that there exists a number $u \in \Reals$ such that
   $P + u \subseteq Q$ if and only if $d_H(A,B) < \eps$.

   First assume that there exists a number $u \in \Reals$ that solves the
   \scp\ instance. It follows immediately that $d_H(A,B) < \eps$.
   Indeed, each point in $P + u$ is contained in $Q$, each point in $Q$ is
   within distance $0.8 \eps$ from $\ell_2$ (or $\ell_3$) which is in $A$,
   and each point in $\ell_1$ or $\ell_4$ (which are in $B$) is within
   distance $0.8 \eps$ from some point in $\ell_2$ or $\ell_3$ (which are
   in $A$), and vice versa.

   Now assume that $d_H(A,B) < \eps$. Since the problem is symmetric about the $x$ axis (around 0), we may assume without loss of generality that $v$, the $y$ component of the translation, is non-negative. We show that $v \leq 0.3 \eps$.  Assume to the contrary that $v > 0.3 \eps$.  Then no point of $\ell_4$ (which is in $B$) has a counterpart point of $A$ within distance less than $1.1 \eps$, which is a contradiction. Hence $v \leq 0.3 \eps$. This means that the closest point of every point of $P$ belongs to $Q$ (since the distance to $\ell_1$ or $\ell_4$ is at least $1.3 \eps$).  Therefore $d_H(P,Q) < \eps$, and in particular, $d_H'(P,Q) < \eps$.  Thus, according to \lemref{1shd-res}, $d_H'(P,Q) = 0$, which, in turn, implies that there exists a number $u \in \Reals$ such that $P + u \subseteq Q$, establishing the claim.
\end{proof}

The fastest known algorithm for computing the minimum Hausdorff distance under translation between two point sets $A$ and $B$, where $|A| = |B| = n$, is due to Huttenlocher \etal \cite{hks-uevsi-93}; its running time is $O (n^3 \log n \, \alpha (n^2))$.  Computing the optimal translational matching between two planar segment sets, each of cardinality $n$, can be done in $O (n^4 \log^3 n)$ time by the parametric-search algorithm of Agarwal \etal \cite{ast-apsgo-94}.  The interested reader is referred to the survey of Alt \etal \cite{ag-rgo-99} for further information on geometric matching.

\section{Conclusions}
\seclab{concl}

In this paper we show that several containment problems (of segment sets
and polygons), and computing the minimum Hausdorff distance under
translation between two planar sets of segments, are \tshard.
We conclude by mentioning a few open problems:

\begin{itemize}
    \item Is the special case of \scp, in which both sets contain only points, \tshard\ too?\footnote{ This version of {\scp{}} was introduced to the authors by S.R. Kosaraju.  }

    \item Are other variants of the minimum Hausdorff distance problem \tshard?

    \item Given a set $S$ of $n$ points in the plane, consider the decision problem of determining whether all the $\binom{n}{2}$ distances induced by $S$ are distinct. Is this problem \tshard?
\end{itemize}

\paragraph*{Acknowledgments.}
The authors thank Jeff Erickson and Alon Efrat for helpful discussions concerning the problems studied in this paper and related problems.

\printbibliography

\end{document}